\documentclass[letterpaper, 10 pt, conference]{ieeeconf}
\IEEEoverridecommandlockouts
\newcommand{\comment}[1]{}

\newcommand{\be}[0]{\begin{equation}}
\newcommand{\ee}[0]{\end{equation}}
\newcommand{\ben}[0]{\begin{equation*}}
\newcommand{\een}[0]{\end{equation*}}
\newcommand{\bena}[0]{\begin{eqnarray*}}
\newcommand{\eena}[0]{\end{eqnarray*}}
\newcommand{\bea}[0]{\begin{eqnarray}}
\newcommand{\eea}[0]{\end{eqnarray}}

\usepackage[english]{babel}
\usepackage{amsthm}
\usepackage[version=4]{mhchem}
\usepackage{siunitx}
\usepackage{array}
\usepackage{longtable}
\usepackage{nomencl}
\usepackage{mathrsfs} 
\usepackage{mathtools} 
\usepackage{optidef}
\usepackage{multicol}
\usepackage{svg}
\usepackage{lipsum}
\usepackage{scalerel}
\usepackage{stackengine}
\usepackage{booktabs}
\usepackage{caption}
\usepackage{enumerate}
\usepackage{cite}
\usepackage{amssymb,amsfonts}
\usepackage{algorithmic}
\usepackage{graphicx}
\usepackage{textcomp}
\usepackage{bm}
\usepackage{subfig}

\theoremstyle{definition}
\newtheorem{definition}{Definition}
\newtheorem{theorem}{Theorem}

\newtheorem{assumption}{Assumption}

\newtheorem{proposition}{Proposition}

\begin{document}
\title{Combinatorial Safety-Critical Coordination of Multi-Agent Systems via Mixed-Integer Responsibility Allocation and Control Barrier Functions} 

\author{Johannes Autenrieb$^{1}$, Mark Spiller$^{1}$, Hyo-Sang Shin$^{2}$, and Namhoon Cho$^{3}$
\thanks{$^{1}$ German Aerospace Center (DLR), Institute of Flight Systems, 38108, Braunschweig, Germany.
(email: \texttt{johannes.autenrieb@dlr.de, mark.spiller@dlr.de})}
\thanks{$^{2}$ 
Cho Chun Shik Graduate School of Mobility, Korea Advanced Institute of Science and Technology (KAIST), Daejeon 34141, Republic of Korea.
(email: \texttt{hyosangshin@kaist.ac.kr}) }
\thanks{$^{3}$ 
Department of Aerospace Engineering, Seoul National University, Seoul 08826, Republic of Korea.
(email: \texttt{namhoon.cho@snu.ac.kr}) }
}

\maketitle

\begin{abstract}
This paper presents a hybrid safety-critical coordination architecture for multi-agent systems operating in dense environments. While control barrier functions (CBFs) provide formal safety guarantees, decentralized implementations typically rely on ego-centric safety filtering and may lead to redundant constraint enforcement and conservative collective behavior. To address this limitation, we introduce a combinatorial coordination layer formulated as a mixed-integer linear program (MILP) that assigns collision-avoidance responsibilities among agents. By explicitly distributing enforcement tasks, redundant reactions are eliminated and computational complexity is reduced. Each agent subsequently solves a reduced local quadratic program enforcing only its assigned constraints.
\end{abstract}

\section{Introduction}
Operating multiple autonomous agents in close proximity introduces challenges beyond nominal guidance and task execution. In many aerospace and robotic applications such as cooperative surveillance, formation flight, and coordinated engagement, agents must achieve mission objectives while maintaining strict safety under tightly coupled dynamics. In these settings, collision avoidance becomes a continuous coordination requirement that directly affects overall system performance. Although high-level assignment and cooperative guidance strategies can align global objectives and target allocations~\cite{LI2024109212,Li2023}, they generally do not resolve conflicts at the trajectory level~\cite{Lv2024}. Consequently, agents with consistent task assignments may still generate incompatible motion plans, making additional safety coordination necessary.

Control Barrier Functions (CBFs) have emerged as an effective tool for enforcing safety constraints in continuous-time systems by guaranteeing forward invariance of admissible sets~\cite{ames2016control,Xiao_2022hocbf}. In multi-agent collision avoidance, CBF-based formulations are commonly deployed in a decentralized manner, where each agent independently enforces all relevant pairwise safety constraints~\cite{BORRMANN2015UAV, autenrieb2025f}. 
While this approach preserves safety under bounded control authority, it induces a rapid growth in the number of active constraints as interaction density increases, which in turn amplifies computational demands and can lead to infeasibility in tightly coupled scenarios~\cite{Aloor_2025,Isaly2024OnTheFeas}.

\begin{figure}[!t]
    \centering
    \includegraphics[width=0.9\columnwidth]{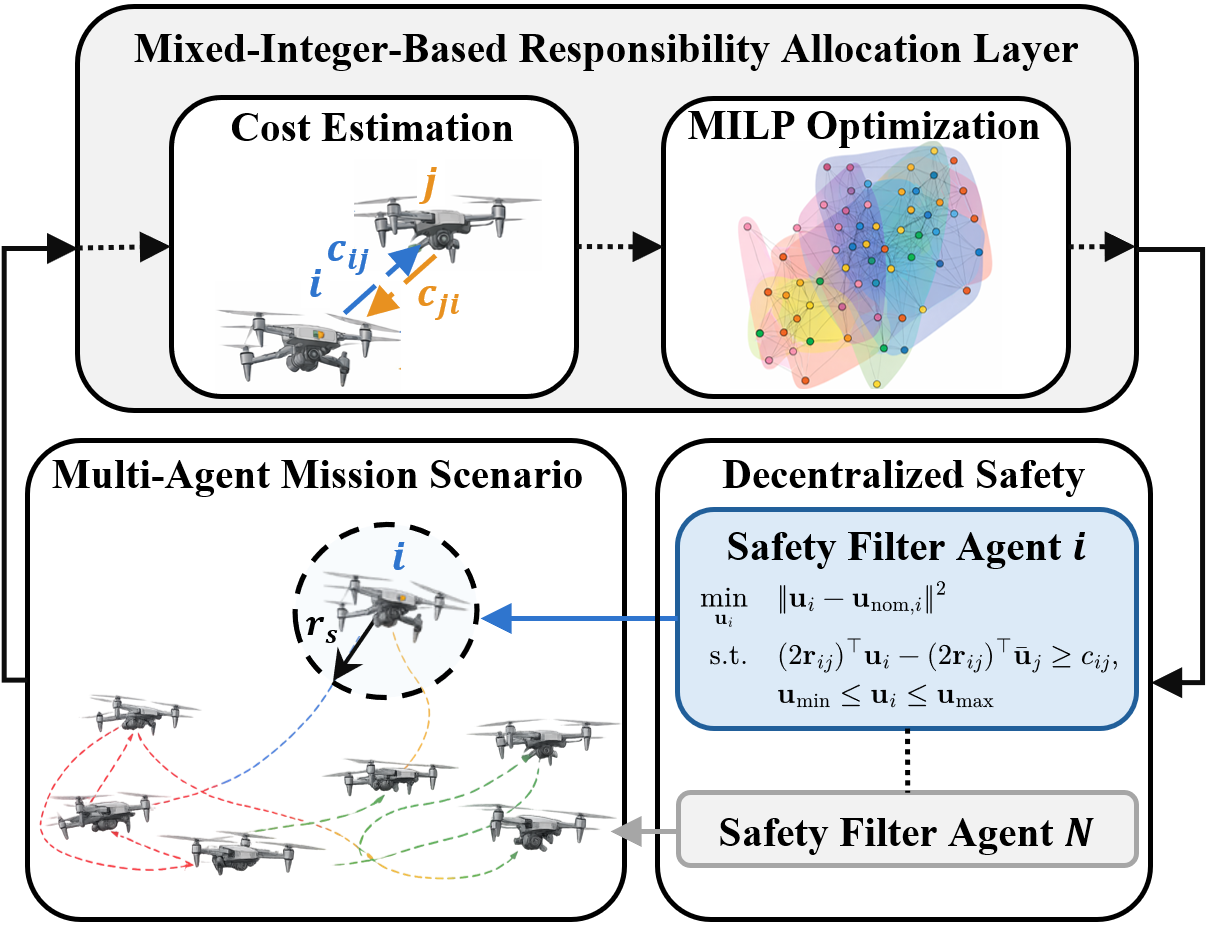}
    \caption{Illustration of the control architecture with mixed-integer-based coordination and decentralized safety filters.}
    \label{fig:overview_concept}
\end{figure}

More fundamentally, most of the existing decentralized safety enforcement methods treat multi-agent collision avoidance as an ego-centric decision problem. Each agent maintains safety by enforcing local safety constraints, adjusting its nominal control input only as much as necessary, while assuming neighboring agents behave similarly. Such locally optimal responses do not generally yield system-efficient outcomes: multiple agents may react unnecessarily to the same interaction, resulting in redundant control effort and degraded collective performance. From a system-level perspective, the objective is not to ensure safety at all costs for each agent independently, but to minimize the aggregate control burden required to maintain safety across the team \cite{autenrieb2025.auction}.

To address this issue, this paper introduces a hybrid control architecture that combines mixed-integer coordination with decentralized safety-critical control (see Fig.~\ref{fig:overview_concept}). Multi-agent collision avoidance is posed as a combinatorial optimization problem, where CBF constraints enforce safety and binary decision variables assign collision-avoidance responsibility among agents. This explicit allocation enables selective enforcement of safety constraints, allowing conflicts to be resolved by agents best positioned to do so. Further, each agent solves a reduced local safety problem, considering only its assigned interactions, which improves the computational effort and reduces unnecessary conservatism.  Numerical results show that this approach preserves formal safety guarantees while significantly increasing the mission performance.
\section{Preliminaries}
\label{sec:Preliminaries}
Consider a control-affine nonlinear system 
\begin{align}
    \dot{\mathbf{x}} &= \mathbf{f}(\mathbf{x}) + \mathbf{g}(\mathbf{x}) \mathbf{u}, \label{NonlinearPlant1}
\end{align}
where \( \mathbf{x} \in \chi \subset R^n\), \( \mathbf{u} \in U \subset R^m \), \( \mathbf{y} \in \mathbb{R}^p \), and \( \mathbf{f}: \chi \to \mathbb{R}^n \), \( \mathbf{g}: \chi \to \mathbb{R}^{n \times m} \), and \( \mathbf{p}: \chi \to \mathbb{R}^p \) are sufficiently smooth functions. To define safety, we consider a continuously differentiable function \( h: \chi \rightarrow \mathbb{R} \) and a set \( S \) defined as the zero-superlevel set of \( h \), yielding:
\begin{equation}
    \label{Safe_set_1}
    S \triangleq \big\{ \mathbf{x} \in \chi \mid h(\mathbf{x}) \geq 0 \big\}.
\end{equation}



We introduce the notion of a {control barrier function} (CBF) such that its existence allows the system to be rendered safe with respect to \( S \) using a control input \( \mathbf{u} \) \cite{ames2016control}.
\begin{definition}[CBF, \cite{Ames_2017}]
Let \( S \subset \chi \) be the zero-superlevel set of a continuously differentiable function \( h: \chi \rightarrow \mathbb{R} \). The function \( h \) is a CBF for \( S \) for all \( \mathbf{x} \in S \), if there exists a class \( \mathcal{K}_{\infty} \) function \( \alpha(h(\mathbf{x})) \) such that for the dynamics defined in \eqref{NonlinearPlant1} we obtain:
\begin{equation}
    \label{ControlBarrierFunction_simple}
    \sup_{\mathbf{u}\in U} [L_{\mathbf{f}} h(\mathbf{x}) + L_{\mathbf{g}} h(\mathbf{x}) \mathbf{u}]  \geq -\alpha(h(\mathbf{x})),
\end{equation}
\end{definition}

\begin{theorem}
\label{theorem_LCBF}
Given a set \( S \subset \chi \), defined via the associated CBF as in \eqref{Safe_set_1}, any Lipschitz continuous controller \( \mathbf{k}(\mathbf{x}) \in K_{S}(\mathbf{x}) \) with 
\begin{equation}
    K_{S} (\mathbf{x}) = \big\{ \mathbf{u} \in U : L_{\mathbf{f}} h(\mathbf{x}) + L_{\mathbf{g}} h(\mathbf{x}) \mathbf{u} + \alpha(h(\mathbf{x})) \ge 0 \big\}
    \label{definition_safe_controller}
\end{equation}
renders the system \eqref{NonlinearPlant1} forward invariant within \( S \) \cite{XU2015}. One way to construct a controller satisfying \eqref{definition_safe_controller} is through a quadratic program-based safety filter, as proposed in~\cite{Ames_2014}: 
\begin{equation}
\begin{aligned}
\mathbf{u} = \arg\min_{\mathbf{u}\in U} \quad
& \|\mathbf{u} - \mathbf{u}^\star\|_2^2 \\
\text{s.t.}\quad
& L_{\mathbf{f}} h(\mathbf{x})
  + L_{\mathbf{g}} h(\mathbf{x})\,\mathbf{u}
  \ge -\alpha\!\left(h(\mathbf{x})\right),
\end{aligned}
\label{eq:qp_cbf}
\end{equation}
where $\mathbf{u}^\star$ is a performance-oriented, but potentially not safe, control input.
\end{theorem}
For safety constraints whose relative degree with respect to the system dynamics
is greater than one, the CBF framework admits a systematic extension in the form
of HOCBFs \cite{Xiao_2022hocbf}. We consider a time-varying function to define an invariant set for system \eqref{NonlinearPlant1}. For a $d$-th order differentiable function $h : \mathbb{R}^n \times [t_0, \infty) \to \mathbb{R}$ (where $t_0$ denotes the initial time), we define a series of functions $\psi_0, \psi_1, \ldots, \psi_d$ recursively as:
\begin{align}
    \psi_0(x(t), t) &= h(x(t), t),\notag\\
    \psi_1(x(t), t) &= \dot{\psi}_0(x(t), t) + \alpha_1(\psi_0(x(t), t)), \notag\\
    \psi_2(x(t), t) &= \dot{\psi}_1(x(t), t) + \alpha_2(\psi_1(x(t), t)), \label{eqn:HOCBF_safe_sets_1} \\
    &\vdots \nonumber \notag\\
    \psi_d(x(t), t) &= \dot{\psi}_{d-1}(x(t), t) + \alpha_d(\psi_{d-1}(x(t), t)), \notag
\end{align}
where $\alpha_1, \alpha_2, \ldots, \alpha_d$ are class $\mathcal{K}$ functions. We further define a series of sets $S_1(t), S_2(t), \ldots, S_d(t)$ associated with \eqref{eqn:HOCBF_safe_sets_1}  in the form:
\begin{align}
    S_1(t) &= \{x \in \mathbb{R}^n : \psi_0(x(t), t) \geq 0\}, \notag\\
    S_2(t) &= \{x \in \mathbb{R}^n : \psi_1(x(t), t) \geq 0\}, \label{eqn:HOCBF_safe_sets_2}\\
    &\vdots \nonumber \\
    S_d(t) &= \{x \in \mathbb{R}^n : \psi_{d-1}(x(t), t) \geq 0\}.\notag
\end{align}
The forward invariance of these sets collectively ensures that the system trajectory remains within a safe region for all $t \geq t_0$.
\begin{definition}[HOCBF, \cite{Xiao_2022hocbf}]
\label{HOCBF_definition}
Let $S_1(t), S_2(t), \ldots, S_d(t)$ be defined by \eqref{eqn:HOCBF_safe_sets_2} and $\psi_0(x(t), t), \psi_1(x(t), t), \ldots, \psi_d(x(t), t)$ be defined by \eqref{eqn:HOCBF_safe_sets_1}.
A function $h : \mathbb{R}^n \times [t_0, \infty) \to \mathbb{R}$ is a HOCBF of relative degree $d$ for the system defined in \eqref{NonlinearPlant1} if there exist differentiable class $\mathcal{K}$ functions $\alpha_1, \alpha_2, \ldots, \alpha_d$ such that:
\begin{equation}
\begin{split}
    L_f^d h(x(t), t) &+ L_g^{d-1} h(x(t), t) u\\& + \frac{\partial^d h(x(t), t)}{\partial t^d} + O(h(x(t), t))\\& + \alpha_d(\psi_{d-1}(x(t), t)) \geq 0,
\end{split}
\end{equation}
for all $(x(t), t) \in S_1(t) \cap S_2(t) \cap \cdots \cap S_d(t) \times [t_0, \infty)$. In the above equation, $O(h(x(t), t))$ denotes the remaining Lie derivatives along $f$ and partial derivatives with respect to $t$ with degree less than or equal to $d - 1$.
\end{definition}

\section{Problem Setup}
\label{sec:Problem}
Let \( \mathcal{M} = \{1, 2, \dots, N\} \) denote the index set of \( N \) effectors. The dynamics of each agent \( i \in \mathcal{M} \) are described by a second-order integrator model in three-dimensional Euclidean space: 
\begin{equation}
    \dot{\mathbf x}_i = 
    \begin{bmatrix}
        \dot{\mathbf p}_i \\ \dot{\mathbf v}_i
    \end{bmatrix}
    =
    \begin{bmatrix}
        0 & \mathbf I \\
        0 & 0
    \end{bmatrix}
    \begin{bmatrix}
        \mathbf p_i \\ \mathbf v_i
    \end{bmatrix}
    +
    \begin{bmatrix}
        0 \\ \mathbf I
    \end{bmatrix}
    \mathbf u_i,
    \label{eq:LTI_system_problem}
\end{equation}
where \( \mathbf{p}_i \in \mathbb{R}^3 \) denotes the position of agent \( i \), \( \mathbf v_i \in \mathcal V  \subset \mathbb{R}^3 \) its velocity, and \( \mathbf u_i \in \mathcal U \subset \mathbb{R}^3 \).

The agents are assumed to be point masses subject to control and physical constraints, including input saturation.
We assume the following physical bounds:


\begin{equation}
\mathcal V := \{ \mathbf v \in \mathbb R^m \mid \mathbf v_{\min} \le \mathbf v \le \mathbf v_{\max} \},
\label{eq:velocity_bound}
\end{equation}

\begin{equation}
\mathcal U := \{ \mathbf u \in \mathbb R^m \mid \mathbf u_{\min} \le \mathbf u \le \mathbf u_{\max} \}.
\label{eq:input_bound}
\end{equation}

Define the relative quantities between two agents $i$ and $j$:
\begin{equation}
    \mathbf r_{ij} = \mathbf p_i - \mathbf p_j, \quad \mathbf v_{ij} = \mathbf v_i - \mathbf v_j.
    \label{eq:realtive_kinematics}
\end{equation}

Each agent \( i \in \mathcal{M} \) is assigned an individual mission objective, such as reaching a target location or performing a task in space. These objectives are assumed to be externally provided and fixed throughout the mission. To achieve its objective, each agent generates a nominal control input \( \mathbf{u}_{\mathrm{nom},i} \in \mathbb{R}^3 \) based on its current state \( x_i \in \mathbb{R}^6 \). This nominal control law may arise from any suitable controller, and is abstractly represented by
\begin{equation}
    \mathbf{u}_{\mathrm{nom},i} = K_i(x_i, t),
\end{equation}
where \( K_i: \mathbb{R}^6 \times \mathbb{R}_{\geq 0} \to \mathbb{R}^3 \) denotes the performance-oriented control strategy designed. However, such nominal controllers do not account for safety-critical constraints arising from shared operation in dynamic multi-agent environments. To ensure safe operation, a real-time safety mechanism is required to modify the nominal input and enforce constraint satisfaction. For large-scale systems with tight coupling and limited communication, this necessitates safety filters that are both formally correct and computationally scalable. 
\section{Combinatorial Safety for Multi-Agent Systems via Mixed-Integer Optimization}
\label{Combinatorial Safety for Multi-Agent Systems via Mixed-Integer Optimization}
Ensuring collision-free operation of large-scale multi-agent systems (MAS) requires mechanisms that go beyond purely ego-centric safety reasoning, since independent safety reactions can lead to redundant control effort and degraded collective performance. We therefore begin by formulating individual safety constraints using HOCBFs under the assumption of pairwise interactions. For each agent \( i \in \mathcal{M} = \{1, \dots, N\} \) and each neighboring agent \( j \in \mathcal{N}_i \subseteq \mathcal{M} \setminus \{i\} \), a HOCBF is constructed to enforce forward invariance of the safe set
\begin{equation}
\mathcal{S}_{ij} = \{(x_i, x_j) \mid h_{ij}(x_i, x_j) \geq 0\}.
\end{equation}
The barrier function
\begin{equation}
h_{ij}(\mathbf{x}_i, \mathbf{x}_j) = \|\mathbf{r}_{ij}\|^2 - r_s^2
\label{eq:barrier}
\end{equation}
ensures a minimum separation radius \( r_s > 0 \). Its time derivatives follow from the relative dynamics,
\begin{equation}
\dot{h}_{ij} = 2 \mathbf{r}_{ij}^\top \mathbf{v}_{ij}, 
\qquad
\ddot{h}_{ij} = 2 \|\mathbf{v}_{ij}\|^2 + 2 \mathbf{r}_{ij}^\top (\mathbf{u}_i - \mathbf{u}_j),
\label{eq:barrier_derivatives}
\end{equation}
where \( \mathbf{u}_i \) and \( \mathbf{u}_j \) denote the control inputs of agents \( i \) and \( j \), respectively. The recursive HOCBF construction reads
\begin{equation}
\psi_0 = h_{ij}, \quad 
\psi_1 = \dot{\psi}_0 + \gamma_1 \psi_0, \quad 
\psi_2 = \dot{\psi}_1 + \gamma_2 \psi_1,
\label{eq:recursive_hocbf}
\end{equation}
with gains \( \gamma_1,\gamma_2>0 \). Satisfaction of \( \psi_2 \ge 0 \) for all \( t>0 \) guarantees forward invariance of \( \mathcal{S}_{ij} \). The resulting affine safety constraint is
\begin{equation}
(2\mathbf{r}_{ij})^\top \mathbf{u}_i 
- (2\mathbf{r}_{ij})^\top \mathbf{u}_j 
\ge c_{ij},
\label{eq:affine_constraint}
\end{equation}
where
\begin{equation}
c_{ij} = -2 \|\mathbf{v}_{ij}\|^2 
- 2 (\gamma_1 + \gamma_2) \mathbf{r}_{ij}^\top \mathbf{v}_{ij}
- \gamma_2 (\|\mathbf{r}_{ij}\|^2 - r_s^2).
\label{eq:cij}
\end{equation}

In a purely decentralized setting, agent \( i \) does not have access to the true control input \( \mathbf{u}_j \) of agent \( j \), so~\eqref{eq:affine_constraint} cannot be enforced directly. To obtain a locally implementable formulation, behavioral assumptions or input predictions must be introduced. For instance, reciprocal avoidance \( \mathbf{u}_j = -\mathbf{u}_i \) may be adopted under cooperative communication, whereas a non-adversarial assumption \( \mathbf{u}_j = \mathbf{0} \) can be used in the absence of coordination. More generally, the unknown neighbor action is replaced by an estimate \( \bar{\mathbf{u}}_j \), yielding the local quadratic program
\begin{equation}
\begin{aligned}
\min_{\mathbf{u}_i} \quad & \|\mathbf{u}_i - \mathbf{u}_{\mathrm{nom},i}\|^2 \\
\text{s.t.} \quad 
& (2\mathbf{r}_{ij})^\top \mathbf{u}_i 
- (2\mathbf{r}_{ij})^\top \bar{\mathbf{u}}_j 
\ge c_{ij}, 
\quad \forall j \in \mathcal{A}_i, \\
& \mathbf{u}_{\min} \le \mathbf{u}_i \le \mathbf{u}_{\max}.
\end{aligned}
\label{eq:local_qp}
\end{equation}
While~\eqref{eq:local_qp} guarantees pairwise safety under standard HOCBF feasibility conditions, its scalability deteriorates in dense environments. In particular, each agent must account for a growing number of simultaneously pairwise constraints, leading to high-dimensional QPs, redundant enforcement, and potential feasibility loss in dense environments. Moreover, since multiple agents may react simultaneously to the same interaction, the resulting control actions are generally not system-optimal and can induce unnecessary aggregate control effort, thereby degrading the overall collective performance of the MAS.

To reduce redundant enforcement and improve scalability, we introduce a global mixed-integer-based coordination layer. We therefore reinterpret each safety constraint $(i,j)\in\mathcal{A}\subset\mathcal{M}\times\mathcal{M}$ as a binary task that must be enforced by at least one of its incident agents.  Let \( z_{ij}\in\{0,1\} \) denote a binary assignment variable, where \( z_{ij}=1 \) indicates that agent \( i \) is responsible for enforcing constraint \( (i,j) \). The coverage condition
\begin{equation}
z_{ij} + z_{ji} \ge 1,
\qquad \forall (i,j)\in\mathcal{A},
\label{eq:coverage}
\end{equation}
ensures that every safety constraint is enforced by at least one agent. In contrast to the decentralized formulation~\eqref{eq:local_qp}, the centralized coordination problem optimizes over all control inputs simultaneously. Hence, the neighbor input \( \mathbf{u}_j \) appears as a decision variable rather than as an assumed quantity. The coupling between agents is therefore resolved directly at the optimization level. The joint allocation-and-control problem can be written as the following mixed-integer nonlinear program:
\begin{align}
\min_{\bm{z},\,\{\mathbf{u}_i\}} \quad 
& \sum_{i=1}^{N} \|\mathbf{u}_i - \mathbf{u}_{\mathrm{nom},i}\|^2 \label{eq:true_minlp}\\
\text{s.t.}\quad 
& z_{ij}
\Big(
(2\mathbf{r}_{ij})^\top \mathbf{u}_i 
- (2\mathbf{r}_{ij})^\top \mathbf{u}_j 
- c_{ij}
\Big)
\ge 0,
 \forall (i,j)\in\mathcal{A}, \notag\\
& z_{ij} \in \{0,1\}, \notag\\
& z_{ij} + z_{ji} \ge 1, \notag\\
& \mathbf{u}_{\min} \le \mathbf{u}_i \le \mathbf{u}_{\max}.\notag
\end{align}

Owing to the presence of binary assignment variables, a quadratic objective, bilinear coupling effects and the restricted admissible control input set of each agent, Problem~\eqref{eq:true_minlp} is inherently difficult to solve. In dense multi-agent scenarios, the induced combinatorial and potentially nonconvex structure gives rise to a computational burden that is frequently incompatible with the real-time requirements of practical multi-agent systems, thereby motivating the introduction of a surrogate mixed-integer linear program (MILP) formulation presented subsequently. Further such responsibility allocation layer is inherently global, as the assignment vector \( \mathbf{z} \) is computed using access to the interaction set \( \mathcal{A} \) via a globally consistent consensus mechanism and broadcasts the decision for safety enforcement to all relevant agents.

To obtain a tractable surrogate of~\eqref{eq:true_minlp}, we introduce the following structural assumptions.

\begin{assumption}[Absence of higher-order assignment coupling]
\label{ass:no_higher_order}
The incremental deviation induced by assigning constraint $(i,j)$ to agent $i$ is assumed to depend on the corresponding pairwise interaction and does not depend on whether agent $i$ is simultaneously assigned any other constraint $(i,k)$.
\end{assumption}

Assumption~\ref{ass:no_higher_order} removes combinatorial coupling between constraints assigned to the same agent and allows the allocation layer to be modeled via pairwise surrogate costs.
Under this assumption, the incremental cost of assigning safety constraint $(i,j)$ to agent $i$ depends only on the local state $(\mathbf{x}_i,\mathbf{x}_j)$, the nominal input $\mathbf{u}_{\mathrm{nom},i}$ of agent $i$ and the input $u_j$ of agent $j$, and is independent of all other assignment decisions. Consequently, the global allocation objective admits the linear surrogate representation
\begin{equation}
\sum_{(i,j)\in\mathcal{A}} J_{ij}^{(i)}\, z_{ij},
\end{equation}
where $J_{ij}^{(i)}$ denotes the isolated one-constraint projection cost and independent of any other binary variable $z_{kl}$


While Assumption~\ref{ass:no_higher_order} justifies linearization of the allocation objective, we still require a quantitative link between the surrogate edge costs and the actual safety-induced control deviation.

\begin{assumption}[Additive surrogate model for safety-induced deviation]
\label{ass:additive_surrogate}
For each agent $i\in\mathcal{M}$ and any assigned constraint set $\mathcal{B}_i \subseteq \mathcal{A}$, define the true deviation
\[
\Delta_i(\mathcal{B}_i)
:=
\min_{\mathbf{u}_i}
\|\mathbf{u}_i - \mathbf{u}_{\mathrm{nom},i}\|^2
\]
subject to all constraints in $\mathcal{B}_i$ and the input bounds. We assume that the induced deviation satisfies the additive upper bound
\begin{equation}
\Delta_i(\mathcal{B}_i)
\le
\sum_{j\in\mathcal{B}_i} J_{ij}^{(i)}.
\label{eq:additive_upper_bound}
\end{equation}
\end{assumption}

While Assumptions~\ref{ass:no_higher_order} and~\ref{ass:additive_surrogate} justify linearization of the allocation objective, directly incorporating the continuous input constraints $\mathbf u_{\min} \le \mathbf u_i \le \mathbf u_{\max}$ of~\eqref{eq:true_minlp} into a combinatorial surrogate remains computationally demanding and may still lead to long run-times and feasibility issues in dense scenarios. To preserve tractability, the surrogate coordination layer focuses exclusively on responsibility allocation, while actuation limits are enforced at the local quadratic program level where safety filtering is ultimately performed. We therefore introduce the following structural feasibility assumption.



\begin{assumption}[Local safety-filter feasibility under assignment]
\label{ass:local_feasibility}
For any responsibility assignment $\bm z$ satisfying the coverage condition~\eqref{eq:coverage}, and for any agent $i\in\mathcal M$, the corresponding local QP-based safety-filter  \eqref{eq:local_qp} is feasible under the imposed input bounds $\mathbf u_{\min}\le \mathbf u_i\le \mathbf u_{\max}$.
\end{assumption}

Assumption~\ref{ass:local_feasibility} formalizes the separation between the discrete allocation layer and the continuous safety-enforcement layer. The allocation mechanism determines responsibility distribution, whereas feasibility of the local safety-filter QPs is a property of the underlying control synthesis under input constraints. Since conventional decentralized schemes already require feasibility of these QPs when enforcing all safety constraints, the proposed allocation modifies only the distribution of enforcement tasks without altering the structural feasibility requirement. 


With the surrogate cost $J_{ij}^{(i)}$ of assigning constraint $(i,j)$ to agent $i$. Under Assumptions~\ref{ass:no_higher_order}--\ref{ass:additive_surrogate}, the coordination layer can be approximated by the following MILP:
\begin{equation}
\begin{aligned}
\min_{\{z_{ij}\}} \quad 
& \sum_{(i,j)\in\mathcal{A}} J_{ij}^{(i)}\, z_{ij} \\
\text{s.t.} \quad 
& z_{ij} \in \{0,1\}, \\
& z_{ij} + z_{ji} \ge 1,
  \quad \forall (i,j)\in\mathcal{A}.
\end{aligned}
\label{eq:milp_surrogate}
\end{equation}

The resulting assignment induces a sparse, directed responsibility graph
\(
\mathcal{G}_S = (\mathcal{M}, \mathcal{E}_S)
\),
with \( \mathcal{E}_S = \{ (i \to j) \mid z_{ij} = 1 \} \), and defines the responsibility-aware sets
\begin{equation}
\mathcal{A}_i^\star := \{ j \in \mathcal{A} \mid z_{ij} = 1 \}.
\label{eq:assigned_constraints}
\end{equation}

To construct the cost coefficients \( J_{ij}^{(i)} \), we recall the affine HOCBF condition~\eqref{eq:affine_constraint} and replace the unknown neighbor input \( \mathbf{u}_j \) with its admissible estimate \( \bar{\mathbf{u}}_j \), yielding
\begin{equation}
(2\mathbf{r}_{ij})^\top \mathbf{u}_i \ge b_{ij}, \qquad
b_{ij} := c_{ij} + (2\mathbf{r}_{ij})^\top \bar{\mathbf{u}}_j.
\label{eq:halfspace_bid}
\end{equation}
For this single constraint, consider the projection problem
\begin{equation}
\min_{\mathbf{u}_i}\ \|\mathbf{u}_i - \mathbf{u}_{\mathrm{nom},i}\|^2
\quad \text{s.t.}\quad (2\mathbf{r}_{ij})^\top \mathbf{u}_i \ge b_{ij},
\label{eq:proj_qp_bid}
\end{equation}
whose solution follows directly from the KKT conditions,
\begin{equation}
\mathbf{u}_i^\star =
\mathbf{u}_{\mathrm{nom},i} +
\frac{\bigl(b_{ij} - (2\mathbf{r}_{ij})^\top \mathbf{u}_{\mathrm{nom},i}\bigr)_+}{\|2\mathbf{r}_{ij}\|^2}
(2\mathbf{r}_{ij}),
\label{eq:proj_closed_bid}
\end{equation}
with $(x)_+:=\max\{x,0\}$ and the associated deviation cost
\begin{equation}
\|\mathbf{u}_i^\star - \mathbf{u}_{\mathrm{nom},i}\|^2 =
\frac{\bigl(b_{ij} - (2\mathbf{r}_{ij})^\top \mathbf{u}_{\mathrm{nom},i}\bigr)_+^2}{\|2\mathbf{r}_{ij}\|^2}.
\label{eq:dev_cost_bid}
\end{equation}
This quantity captures the minimal corrective effort required for agent \( i \) to enforce constraint \( (i,j) \) in isolation, and under homogeneous assumptions serves as a consistent proxy for the assignment cost \( J_{ij}^{(i)} \). Accordingly, we define the local avoidance cost
\begin{equation}
J_{ij}^{(i)} :=
\frac{\bigl(b_{ij} - (2\mathbf{r}_{ij})^\top \mathbf{u}_{\mathrm{nom},i}\bigr)_+^2}{\|2\mathbf{r}_{ij}\|^2},
\label{eq:bid_value}
\end{equation}
with \( J_{ij}^{(i)} = +\infty \) if the corresponding control violates input bounds. Conditioned on the global assignment \( \mathbf{z} \) computed by \eqref{eq:milp_surrogate}, each agent enforces only its assigned constraints by solving the reduced local safety filter
\begin{equation}
\begin{aligned}
\min_{\mathbf{u}_i} \quad & \|\mathbf{u}_i - \mathbf{u}_{\mathrm{nom},i}\|^2 \\
\text{s.t.} \quad 
& (2\mathbf{r}_{ij})^\top \mathbf{u}_i - (2\mathbf{r}_{ij})^\top \bar{\mathbf{u}}_j \ge c_{ij},
\quad \forall j \in \mathcal{A}_i^\star, \\
& \mathbf{u}_{\min} \le \mathbf{u}_i \le \mathbf{u}_{\max}.
\end{aligned}
\label{eq:reduced_qp_style}
\end{equation}

This allocation-aware filtering scheme eliminates redundant enforcement and reduces local computational complexity in dense or large interaction scenarios. 

\begin{theorem}[Global safety in centralized responsibility allocation]
\label{thm:global_safety_allocation}
Consider the MAS governed by the HOCBF constraints~\eqref{eq:barrier}--\eqref{eq:cij}. Suppose that Assumption~\ref{ass:local_feasibility} holds and that the neighbor input estimates used in~\eqref{eq:halfspace_bid} satisfy $(2\mathbf r_{ij})^\top \bar{\mathbf u}_j \ge (2\mathbf r_{ij})^\top \mathbf u_j$ for all $(i,j)\in\mathcal A$. Then the set
\[
\mathcal{S} := \bigcap_{(i,j)\in\mathcal{A}} \mathcal{S}_{ij}
\]
is forward invariant under the closed-loop dynamics induced by~\eqref{eq:reduced_qp_style}. Consequently, the overall MAS remains collision-free for all \( t \ge 0 \).
\end{theorem}

\begin{proof}[Proof]
Fix any pair \( (i,j) \in \mathcal{A} \). By the coverage condition~\eqref{eq:coverage}, either \( z_{ij} = 1 \) or \( z_{ji} = 1 \). Without loss of generality, assume \( z_{ij} = 1 \), so that agent \( i \) enforces the corresponding constraint in~\eqref{eq:reduced_qp_style}. By construction of~\eqref{eq:halfspace_bid}, satisfaction of this constraint implies the original affine HOCBF condition~\eqref{eq:affine_constraint}. Standard HOCBF arguments then guarantee forward invariance of \( \mathcal{S}_{ij} \). Since the argument applies to every pair and the assignment is global, forward invariance of \( \mathcal{S} \) follows.
\end{proof}

While Theorem~\ref{thm:global_safety_allocation} establishes that centralized responsibility allocation preserves global collision avoidance, it does not address how the distribution of enforcement tasks affects control performance. We therefore next show that the proposed responsibility assignment minimizes safety-induced control effort.

\begin{proposition}[Minimal additive upper bound under responsibility allocation]
\label{prop:minimal_invasiveness_clean}
Suppose Assumption~\ref{ass:additive_surrogate} holds.
Let $\bm z^\star$ be an optimal solution of~\eqref{eq:milp_surrogate}.
Then, among all feasible assignments satisfying the coverage constraints, $\bm z^\star$ minimizes an additive upper bound on the total
safety-induced deviation, i.e.,
\begin{equation}
\sum_{i\in\mathcal{M}}
\|\mathbf{u}_i^\star
- \mathbf{u}_{\mathrm{nom},i}\|^2
\le
\sum_{(i,j)\in\mathcal{A}}
J_{ij}^{(i)}\, \tilde z_{ij}
\label{eq:proxy_bound_clean}
\end{equation}
for any alternative feasible assignment $\tilde{\bm z}$.
\end{proposition}

\begin{proof}[Proof]
Fix any feasible assignment $\tilde{\bm z}$. For each agent $i$, let $\mathcal{B}_i(\tilde{\bm z}) := \{ j\in\mathcal{A}_i \mid \tilde z_{ij}=1\}$. By Assumption~\ref{ass:additive_surrogate}, \[ \|\mathbf{u}_i^\star(\tilde{\bm z})  - \mathbf{u}_{\mathrm{nom},i}\|^2 = \Delta_i(\mathcal{B}_i(\tilde{\bm z})) \le \sum_{j\in\mathcal{B}_i(\tilde{\bm z})} J_{ij}^{(i)}. \] Summing over all agents yields \[ \sum_{i\in\mathcal{M}} \|\mathbf{u}_i^\star(\tilde{\bm z})  - \mathbf{u}_{\mathrm{nom},i}\|^2 \le \sum_{(i,j)\in\mathcal{A}} J_{ij}^{(i)}\, \tilde z_{ij}. \] Since $\bm z^\star$ minimizes the right-hand side of \eqref{eq:proxy_bound_clean}, the claim follows.
\end{proof}

\section{Results}
Numerical simulations were conducted in MATLAB to evaluate the proposed coordination mechanisms considering the problem presented in Section~\ref{sec:Problem}. The fully decentralized local QPs as defined in Eq.~\eqref{eq:local_qp} and full MILP coordination as defined in Eq.~\eqref{eq:milp_surrogate} with local QPs are compared. A system of $N=100$ agents is considered. Initial and goal positions are sampled uniformly in $[-5.2,5.2]^2$. The safety radius is $r_s=0.3$. Safety is enforced via local QP-based safety filters following Eq.~\eqref{eq:reduced_qp_style} and the defined HOCBF approach with the parameters $(\gamma_1,\gamma_2)=(5,2)$. The velocity and input bounds used in activation and constraint enforcement are $v_{\max/\min}=\pm 5$, while for the sake of property analysis input constraints are neglected. Figure~\ref{fig:trajectoties1}  and Figure~\ref{fig:trajectoties2} show agent trajectories and corresponding distance-to-goal profiles for both compared concepts. The decentralized approach exhibits strong oscillatory paths due to multiple simultaneously active safety enforcements. This results in to a slower convergence, with mission completion at $22.60\,\mathrm{s}$, where mission completion is considered when all agents reached their goal location. The MILP coordination, as proposed in Section \ref{Combinatorial Safety for Multi-Agent Systems via Mixed-Integer Optimization}, reduces redundant constraint enforcement and yields smoother trajectories, reducing mission time to $7.50\,\mathrm{s}$.
\begin{figure}
    \centering
    \includegraphics[width=0.7\columnwidth]{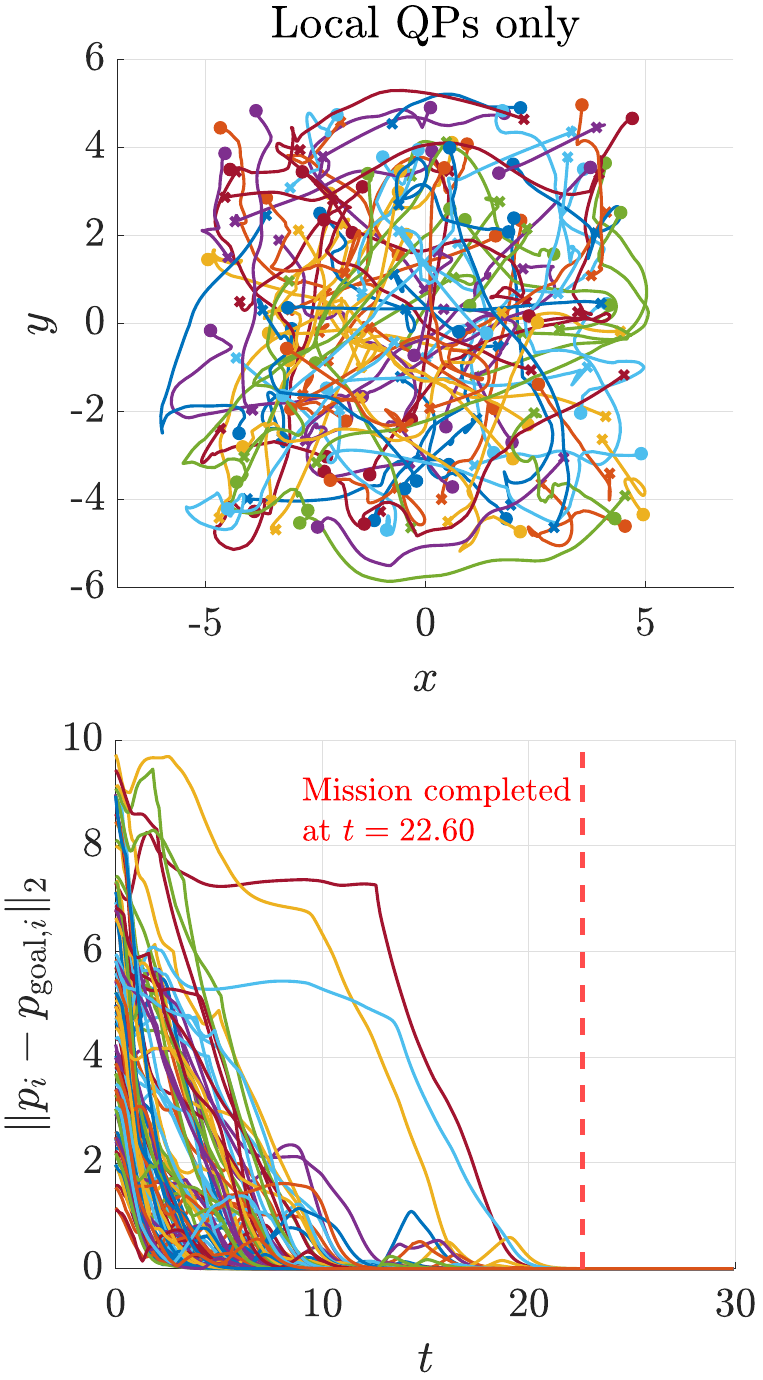}
    \caption{Agent trajectories and distance-to-goal profiles for MAS with fully decentralized QPs.}
    \label{fig:trajectoties1}
\end{figure}
\begin{figure}
    \centering
    \includegraphics[width=0.7\columnwidth]{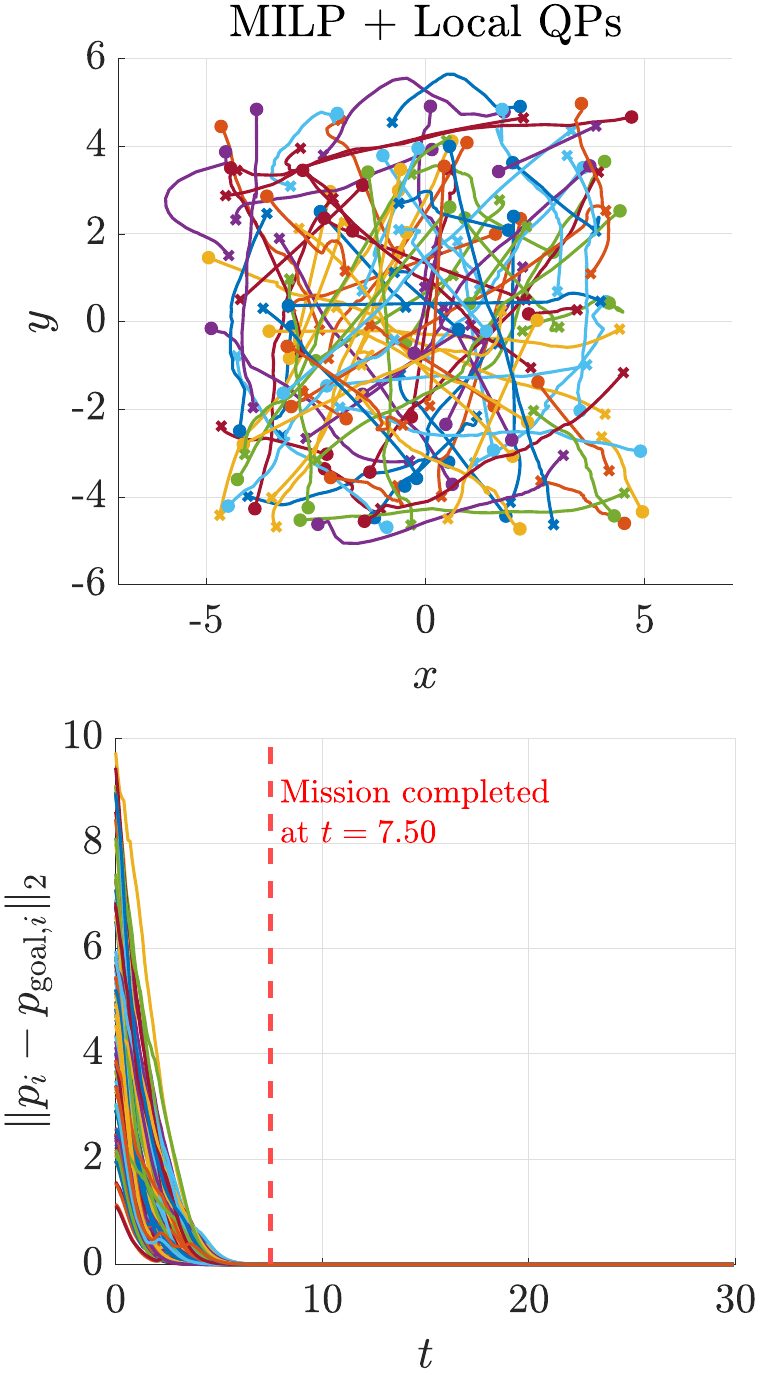}
    \caption{Agent trajectories and distance-to-goal profiles for MAS with MILP coordination and decentralized QPs.}
    \label{fig:trajectoties2}
\end{figure}
Figure~\ref{fig:performance_N30} reports the total deviation cost $\sum_i \|u_i-u_{nom,i}\|_2^2$, the average barrier value $\bar{h}$, and the average local QP execution time $\bar{t}_{\mathrm{QP}}$. The decentralized strategy yields the higher cumulative cost. MILP coordination reduces this cost. The average barrier value confirms reduced conservatism under coordinated strategies while preserving safety. Local QP computation time follows the same trend, with the MILP approach exhibiting for most of the time the lower average execution time.
\begin{figure}
    \centering
    \includegraphics[width=0.8\columnwidth]{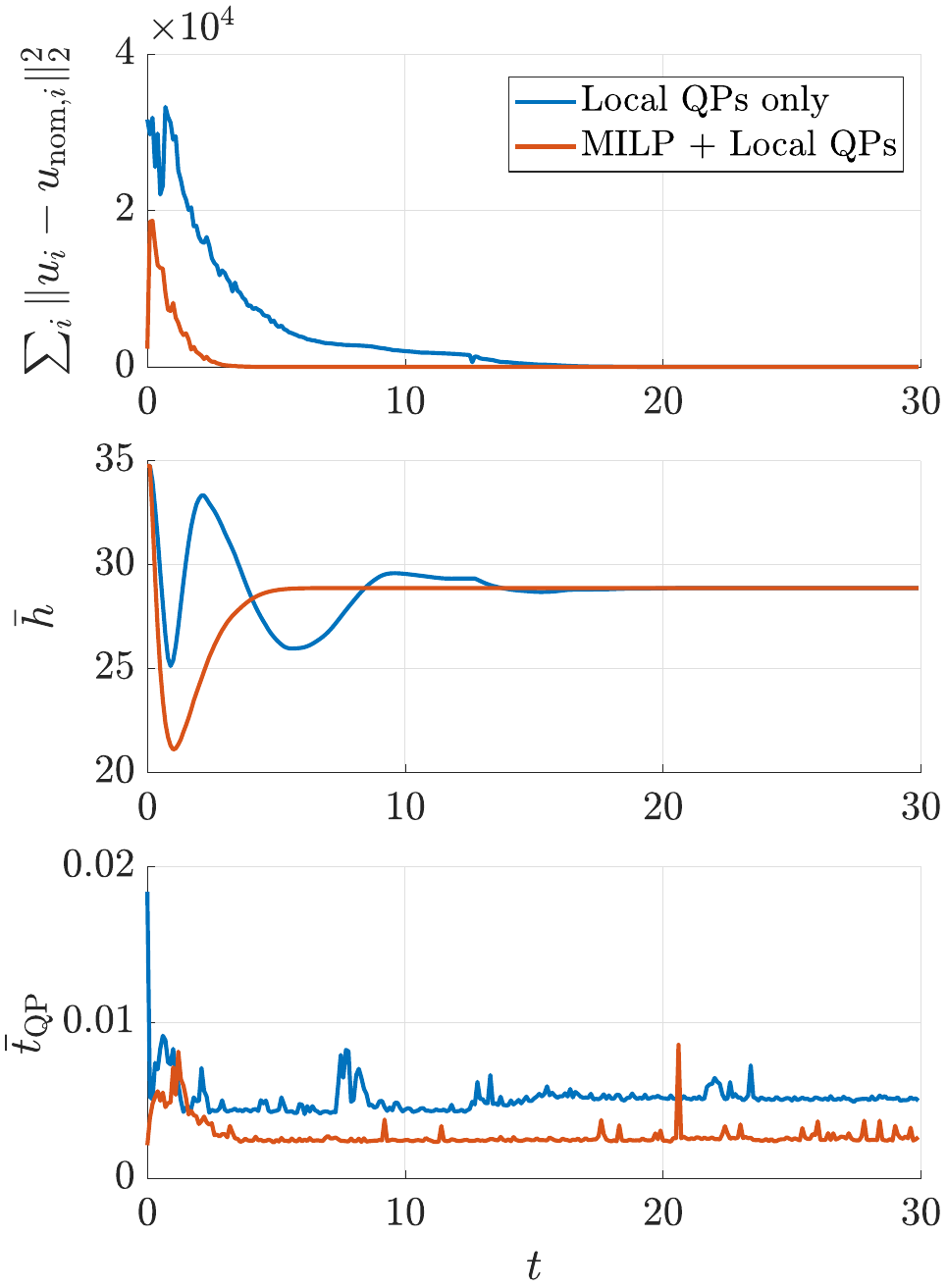}
    \caption{Total deviation cost, average barrier value $\bar{h}$, and average local QP time $\bar{t}_{\mathrm{QP}}$.}
    \label{fig:performance_N30}
\end{figure}

\section{Conclusion}
This paper introduced a combinatorial safety-critical coordination framework that augments decentralized CBF-based filtering with mixed-integer responsibility allocation. By explicitly assigning collision-avoidance responsibilities, redundant pairwise enforcement is eliminated and aggregate control effort is reduced. The separation between discrete coordination and continuous safety filtering preserves formal forward-invariance guarantees while improving scalability in dense multi-agent scenarios. 
\bibliographystyle{IEEEtran}
\bibliography{./references}

\end{document}